\documentclass[11pt]{amsart}
\usepackage{amsmath,amssymb,amsthm,amscd}
\usepackage{url,enumerate}
\usepackage[bbgreekl]{mathbbol}
\usepackage{algorithm}
\usepackage{algorithmic}
\usepackage[dvipsnames]{xcolor}
\definecolor{linkcolor}{rgb}{0.65,0,0}
\definecolor{citecolor}{rgb}{0,0.65,0}
\definecolor{urlcolor}{rgb}{0,0,0.65}
\usepackage[colorlinks=true, linkcolor=linkcolor, urlcolor=urlcolor, citecolor=citecolor]{hyperref}

\theoremstyle{plain} 
\newtheorem{theorem}{Theorem} 
\newtheorem{property}{Property} 

\newtheorem{proposition}[theorem]{Proposition}

\newtheorem{corollary}[theorem]{Corollary}

\theoremstyle{definition}
\newtheorem{definition}[theorem]{Definition}

\theoremstyle{remark}

\numberwithin{theorem}{section}

\newcommand{\TITLE}{Algebraic aspects of solving Ring-LWE, including ring-based improvements in the Blum-Kalai-Wasserman algorithm}

\newcommand{\DATE}{\today}

\newcommand{\CC}{\mathbb{C}}
\newcommand{\FF}{\mathbb{F}}

\newcommand{\QQ}{\mathbb{Q}}
\newcommand{\RR}{\mathbb{R}}
\newcommand{\ZZ}{\mathbb{Z}}

\newcommand{\MOD}[1]{~(\textup{mod}~#1)}
\renewcommand{\pmod}{\MOD}

\newcommand{\ord}{\operatorname{ord}}

\title{\TITLE}

\author{Katherine E. Stange}

\date{\DATE}
\address{%
Department of Mathematics, University of Colorado,
Campux Box 395, Boulder, Colorado 80309-0395}
\email{kstange@math.colorado.edu}

\thanks{This research was supported by NSF-CAREER CNS-1652238 and NSF EAGER DMS-1643552.}

\newcommand{\Fq}{\mathbb{F}_q}
\newcommand{\Fqk}{\mathbb{F}_{q^k}}

  \newcommand{\Rq}{R_q}
 \newcommand{\Sq}{S_q}

\begin{document}

\maketitle

\begin{abstract}
We provide a reduction of the Ring-LWE problem to Ring-LWE problems in subrings, in the presence of samples of a restricted form (i.e.\ $(a,b)$ such that $a$ is restricted to a multiplicative coset of the subring).  To create and exploit such restricted samples, we propose Ring-BKW, a version of the Blum-Kalai-Wasserman algorithm which respects the ring structure.  Off-the-shelf BKW dimension reduction (including coded-BKW and sieving) can be used for the reduction phase.  Its primary advantage is that there is no need for back-substitution, and the solving/hypothesis-testing phase can be parallelized.  We also present a method to exploit symmetry to reduce table sizes, samples needed, and runtime during the reduction phase.  The results apply to two-power cyclotomic Ring-LWE with parameters proposed for practical use (including all splitting types).
\end{abstract}

\section{Introduction}

Ring Learning with Errors (Ring-LWE) \cite{lpr10} \cite{lpr13}, and Learning with Errors (LWE) \cite{reg05} more generally, are leading candidates for post-quantum cryptography.  The cryptographic hard problem (\emph{Search Ring-LWE}) is formally similar to discrete logarithm problems, so that protocols can be transferred from the latter context to the former.  But it also allows for new applications, such as homomorphic encryption \cite{rlwe-homo}.  Ring-LWE is also fortunate in having security reductions from other lattice problems. 

Ring-LWE is distinguished from Learning with Errors (LWE) by the use of lattices from number fields.  This injection of number-theoretical structure leads to performance improvements, but may add vulnerabilities.  So far, the number-theoretical structure has been only weakly exploited for attacks.  The ring structure plays a role in security when the error distribution is skewed \cite{civ} \cite{siam} \cite{cls-galois} \cite{eisentrager2014weak} \cite{elos2015weak}, or the secret is chosen from a subring or other ring-related non-uniform distribution \cite{order-lwe}.  In the related NTRU cryptosystem, the norm and trace maps to subfields play a role in attacks \cite{AlbrechtBaiDucas, CheonJeongLee, GentrySzydlo, KirchnerFouque}. 

However, the best known attacks on Ring-LWE parameters suggested for implementation are still generic attacks for LWE, e.g.\ \cite{newhope}.  The Blum-Kalai-Wasserman (BKW) algorithm is one such attack, which proceeds (in the first phase) combinatorially to create new samples in a linear subspace of the original problem, while controlling error expansion \cite{bkw}.  BKW has the drawback of requiring exponentially many samples, unless sample amplification is used \cite{1222}.  Nevertheless, its performance has been of significant interest:  for analysis and recent improvements, see \cite{albrecht-bkw} \cite{better} \cite{sieve-bkw} \cite{asympt} \cite{coded-bkw} \cite{improved-bkw}.  (Note that sample amplification does not immediately transfer from the LWE to the Ring-LWE setting, at least if one wishes the amplified samples to have Ring-LWE, and not just underlying LWE, format; the analogue would be the `sample rotation' described below.)

This paper focuses on two-power-cyclotomic unital (but equivalently, dual \cite{challenges} \cite{hownotto}) Search Ring-LWE, with no restriction on the splitting behaviour of the prime $q$.  The core of the paper is a reduction from higher-dimensional Ring-LWE problems with samples of a restricted form, to lower-dimensional Ring-LWE problems with the same error width, which is given in Theorem \ref{thm:reduction}.  The restricted form is as follows:  samples $(a,b)$ such that $a$ lies in a cyclotomic subring, or a fixed multiplicative coset of such a subring.  In the context of these theorems, it is natural to ask about creating samples of this restricted form using a ring variant of the Blum-Kalai-Wasserman algorithm.

One thus obtains a \emph{Ring-BKW} algorithm, which uses the reduction phase of BKW, including all known speedups, to reduce the Ring-LWE problem to a subring.  Then, the symmetry of the ring structure allows us to engineer an entire suite of subring problems in polynomially more time, whose solutions collectively solve the original Ring-LWE problem, again in polynomial time.  \emph{Thus, the `hypothesis testing' phase of BKW is parallelized, and the exponential `back-substitution' phase is eliminated (Theorem \ref{thm:reduction}).}  State-of-the-art off-the-shelf code for the BKW reduction phase and hypothesis testing phase may be used.  Note that the reduction phase of BKW is the dominant phase for runtime, and hypothesis testing is typically polynomial, but the now-eliminated back-substitution phase runs in time which is also exponential, but differs only by a smaller polynomial factor from the reduction phase; hence the overall runtime savings is a polynomial factor.  In Section \ref{sec:ring-bkw-alg}, we describe the Ring-BKW algorithm.  

The paper also addresses the use of symmetry to reduce the table sizes in BKW, here termed \emph{advanced keying} in Section \ref{sec:advanced-keying}.  Compared to a BKW reduction phase completely blind to the ring structure, this reduces the table size and samples needed by a factor of the block size, as well as reducing runtime, but requires that block sizes be taken to be a (possibly varying) power of $2$. 

We also discuss a square-root speedup over exhaustive search (which may be used, for example, in hypothesis testing); see Corollary \ref{corollary:square-root}.

See Section \ref{sec:relevance} for more discussion of practical runtime.

The key theoretical properties which are potentially advantageous (to an attacker) of Ring-LWE vs. plain LWE, are:

\begin{enumerate}
  \item Ring homomorphisms into smaller instances of the problem (the main tool of \cite{siam} \cite{cls-galois} \cite{eisentrager2014weak} \cite{elos2015weak}).
  \item The ability to \emph{rotate} samples, e.g.\ replacing $(a,b)$ with $(\zeta a, \zeta b)$ or $(a, \zeta b)$, which are different but related Ring-LWE samples (see notation in Section \ref{sec:setup}); these represent symmetries of the lattice (previously used in lattice sieving \cite{880} \cite{Schneider-Sieve}; more generally, manipulation of samples by multiplication was exploited in \cite{Bunny}).
  \item The existence of subrings as linear subspaces (which is important in \cite{order-lwe}).
  \item More generally, the multiplicative structure of certain linear subspaces.
  \item In the case of 2-power cyclotomics, the orthogonality of the lattice of the ring of integers and the orthogonal nature of the trace.
\end{enumerate}

For us, all five of these attributes play an important role.  It is a secondary purpose of this paper to lay out these advantages in a clear manner, to facilitate future analysis of the security of ring aspects of Ring-LWE.  See Section \ref{sec:advantages}.

Finally, it is also a secondary purpose of this paper to provide a treatment of the Ring-LWE problem which is inviting to the mathematical community.  

Code demonstrating the correctness of the algorithm is available at:

\url{https://math.katestange.net/code/ring-bkw/}.

\subsection*{Acknowledgements}  First, I would like to thank the anonymous referees on an earlier draft of this paper, who pointed out an important simplification.  Second, I would like to thank my mother, Ursula Stange, and my husband, Jonathan Wise, without whose childcare help in the face of snowstorms, viruses, cancellations and fender-benders, this paper simply would not have been completed.  To mathematician moms (and dads) everywhere:  take heart.

\section{Background and Setup for Ring-LWE}
\label{sec:setup}

It is typical to set notation for Ring-LWE as in, for example, \cite{order-lwe}; here we briefly review this notation in our context, and define the Ring-LWE problems.

\subsection{Number field $K$ and ring $R$}
Let $K$ be a number field over the rationals, of degree $n$. 
Then $K$ is equipped with a bilinear form given by a modification of the trace pairing,
\begin{equation}
  \label{eqn:trace-pairing}
  \langle \alpha, \beta \rangle = 
  \sum_{\sigma \RR} \sigma(\alpha\beta)
  +
   \frac{1}{2} \sum_{\sigma \CC} \operatorname{Re}(\sigma(\alpha)\overline{\sigma(\beta)}).
\end{equation}
Here the sums are over real and complex embeddings, respectively (note that including both elements of each pair of conjugate complex embeddings necessitates the factor of $\frac{1}{2}$).  This gives an isomorphism of $K_\RR := \RR \otimes_{\QQ} K$ with $\RR^n$, taking the pairing above to the standard inner product, and \eqref{eqn:trace-pairing} is chosen in such a way that the isomophism is exactly that arising from the \emph{Minkowski} or \emph{canonical} embedding of algebraic number theory.  We can also denote the norm by $||\mathbf{x}|| = \sqrt{\langle \mathbf{x}, \mathbf{x} \rangle}$.  

The ring of integers $R$ of $K$ forms a lattice in $K_\RR$.

\subsection{Gaussian distribution}

Having geometry (in particular a norm $|| \cdot ||^2$) on $K_\RR$ allows us to define Gaussian distributions.  For a \emph{Gaussian parameter} $r > 0$, we write
\[
\rho_r : K_\RR \rightarrow (0,1], \quad 
\rho_r(\mathbf{x}) = \exp( -\pi ||\mathbf{x}||^2 / r^2 ).
\]
Normalizing this to obtain a probability distribution function $r^{-n} \rho_r$, we obtain the \emph{continuous Gaussian probability distribution of width $r$} on $K_\RR$, denoted $D_r$.  

Note that, when considered with respect to an orthonormal basis, such a distribution is the sum of independent distributions in each coordinate, each having width $r$.  In this paper, we are concerned exclusively with this case.

With this normalization, the variance is $r^2/2\pi$, and one standard deviation is $r/\sqrt{2\pi}$.  It is a sum of independent Gaussians in each coordinate for which the range $[-r,r]$ corresponds to $\sqrt{\pi/2} \sim 1.25\ldots$ standard deviations.  

In practice, the tails of the Gaussian may be cut off, so that the number of possible values in each coordinate is finite.  

One may discretize a Gaussian distribution to obtain a distribution $\mathcal{D}_r$ on a lattice $\mathcal{L} \subset K_\RR$.  That is, one takes
\[
  \rho_r(\mathcal{L}) = \sum_{\lambda \in \mathcal{L}} \rho_r(\lambda)
\]
and one samples element $\lambda \in \mathcal{L}$ with probability
\[
  \frac{ \rho_r(\lambda) }{ \rho_r(\mathcal{L}) }.
\]
If $\mathcal{L}$ has an orthonormal basis, then again this distribution consists of independent distributions on the coefficients of the basis.

    \subsection{Prime $q$ and quotient ring $R_q$}

    Let $qR$ be the ideal generated by $q$ in $R$. The fundamental setting of the Ring-LWE problem is the ring $R_q := R/qR$.
    
    Letting $q = \mathfrak{q}_1^{e_1} \cdots \mathfrak{q}_g^{e_g}$ be the unique decomposition of $q$ into distinct prime ideals $\mathfrak{q}_i$ in $R$, the Chinese remainder theorem gives
\[
  R_q \cong \bigoplus_{i=1}^g R/\mathfrak{q}_i^{e_i}.
\]
If $q$ is unramified (which is typically the case), then $e_i = 1$ for all $i$.  If $K$ is Galois (also typically the case in the cryptographic setting), then the Galois group acts transitively on the $\mathfrak{q}_i$ and they all have the same residue degree (the residue degree is the dimension of the quotient field $R/\mathfrak{q}_i$ as an $\FF_q$-vector space).

    \subsection{Ring-LWE distributions}

    For any $s \in R_q$ (the \emph{secret}), and any distribution $\psi$ over $R_q$ (the \emph{error distribution}), we write $A_{s,\psi}$ for the associated \emph{Ring-LWE distribution for secret $s$} over $R_q \times R_q$, given by sampling $a$ uniformly over $R_q$, sampling $e$ from $\psi$, and outputting $(a, b := as + e)$.  

    Such outputs $(a,b)$ are called \emph{samples}, and in a crytographic application, these are observed publicly, while the secret is not meant to be exposed.

For the error distribution, we wish to define a `small' distribution on $R_q$, i.e.\ concentrated near the origin (in comparison to $q$, which is large).  It is typical to choose for the error distribution a discretized Gaussian distribution as described above (considered post factum modulo $qR$).
    This is the context in which security reductions apply.  In implementations, it is sometimes suggested to approximate this by a uniform distribution on a box around the origin, etc.
    
    \subsection{Ring-LWE problems}

    The two fundamental Ring-LWE problems are (a) \emph{search}: to compute the secret, upon observing sufficiently many samples; or (b) \emph{decision}:  to determine if the samples are hiding a secret at all, as opposed to being random noise.  We state them more formally as follows.

    \begin{definition} The \emph{search Ring-LWE problem}, for error distribution $\psi$ and secret distribution $\varphi$, is as follows:  Given an error distribution $\psi$ over $R_q$ and a secret distribution $\varphi$ over $R_q$, and some number of samples drawn from the distribution $A_{s,\psi}$ for some fixed $s$ drawn from $\varphi$, compute $s$.
    \end{definition}

    \begin{definition} The \emph{decisional Ring-LWE problem}, for error distribution $\psi$ and secret distribution $\varphi$, is as follows: Given an error distribution $\psi$ over $R_q$ and a secret distribution $\varphi$ over $R_q$, distinguish with non-negligible advantage, between
      \begin{enumerate}
	\item samples drawn from the distribution $A_{s,\psi}$ for some fixed $s$ drawn from $\varphi$; and 
	\item samples drawn uniformly from $R_q \times R_q$.
      \end{enumerate}
    \end{definition}

    We remark that Ring-LWE is frequently defined in the context of the dual $R^\vee$ (the inverse of the different ideal).  However, in the case that $K$ is a $2^N$-th cyclotomic field, $R \cong 2^{N-1}R^\vee$ and this isomorphism is realized as a scaling in the canonical embedding, and thus preserves the error distribution up to scaling, so we can interchange the dual version with the simpler `unital' version considered here \cite{lpr10}.

    Search-to-decision reductions are known in a variety of contexts \cite{lpr10}.  This paper concerns both problems, but especially the search problem.

    The Ring-LWE problem is formally similar to the discrete logarithm problem, which could be phrased in terms of \emph{samples} $(a, a^s)$ in a finite field:  given $(a, a^s)$, find $s$.  In the ring $R_q$, solving for $s$ given $(a, as)$ can be accomplished using linear algebra (Gaussian elimination), or by multiplication by $a^{-1}$ in the ring.  By introducing a small error $e$, so we have $(a, as+e)$, multiplication by $a^{-1}$ is no longer helpful, and Gaussian elimination becomes useless, as it amplifies the errors to the point of washing out all useful information.  From another perspective, the security stems from the fact that addition of an error value is somehow unpredictably mixing with respect to the multiplicative structure.

    Another consequence of this setup is that given just one sample $(a, b)$, one has as many solutions $s$ to $b = as + e$ as there are possible values for $e$.  In fact, the problem only has a unique solution once we have enough samples.  If the samples are not Ring-LWE samples at all, then with sufficiently many samples, it becomes overwhelmingly likely that there are no values of $s$ so that $b_i - a_is$ is in the support of the error distribution for all samples $s$.  If the samples are Ring-LWE, this is the point at which the true secret is the only solution, with overwhelming probability.

    \section{Specializing to $2$-power cyclotomic Ring-LWE}
    \label{sec:notation}

 We will now specialize to the $2$-power cyclotomic case, fixing values for the variables
 \[
   K, R, q, R_q, n 
 \]
 from the last section, and defining 
 \[
   m, \zeta:= \zeta_m, \chi, \chi_0, E_{\chi_0}
 \]
 for the $2$-power cyclotomic case.  Whenever we say refer to \emph{$2$-power cyclotomic Ring-LWE}, we refer to all the conventions in this section.

 \subsection{Ring $R$}
 We let $K$ and $R$ be the $2n$-th cyclotomic field and ring of integers, respectively, where $n$ is a power of two.  This is of dimension $n$ (note that $\varphi(2n)=n$), and can be presented as
 \[
   R = \ZZ[\zeta_{2n}] = \ZZ[x]/(x^n+1).
 \]
 We will use the notation $m = 2n$ and $\zeta_m$ for a primitive $m$-th root of unity in $R$ \emph{and for its image in quotients of this ring}.  

 \subsection{The $\zeta$-basis for $R$ and its quotients}
 
 A basis for $R$ is
 \[
   1, \zeta_m, \zeta_m^2, \ldots, \zeta_m^{n-1}.
 \]
 This will be called the \emph{$\zeta$-basis}.  We have the relation $\zeta_m^n + 1 = 0$ in $R$ and in all its quotients (this is the $2n$-th cyclotomic polynomial evaluated at $\zeta_m$), but the minimal polynomial for $\zeta_m$ varies in these quotients, and may be a proper divisor of this cyclotomic polynomial.  Nevertheless, in all quotients of $R$, we still obtain a $\zeta$-basis, i.e.\ a power basis in terms of $\zeta := \zeta_m$.

 \subsection{Prime $q$}
 Let $q$ be an odd prime, unramified in $R$.

 \subsection{Ring $R_q$ and further quotients}
 We consider the quotient ring
 \[
   R_q = R/qR \cong (\ZZ/q\ZZ)[x]/(x^n+1), 
 \]
 which is an $\Fq$-vector space of dimension $n$.  We may use the same $\zeta$-basis for this ring (to be explicit, the images of the $\zeta$ basis for $R$ under the reduction modulo $q$).
 
 We may also consider further quotients $R/\mathfrak{a}$ for $\mathfrak{a} \mid qR$.  We may also use a $\zeta$-basis for these rings, although it may be of lower dimension over $\Fq$ (so fewer powers required).  We have
   \[
     R/\mathfrak{a} \cong \Fq[x]/(g(x))
   \]
   where $g(x) \mid x^n+1$.  In particular, identifying $\zeta \in R/qR$ with its image in $R/\mathfrak{a}$, the latter has an $\Fq$-basis $1, \zeta, \zeta^2, \ldots, \zeta^{\deg(g)-1}$.

   \subsection{Error distribution $\chi$, coefficient distribution $\chi_0$ and coefficient support $E_{\chi_0}$}
   \label{sec:error}

   We will denote the error distribution by $\chi$.  If this error distribution is formed using independent identically distributed coefficients on the $\zeta$-basis, with coefficient distribution $\chi_0$ supported on a subset $E_{\chi_0} \subseteq \Fq$, then we say that $\chi$ is \emph{formed on the $\zeta$-basis with coefficients distributed according to $\chi_0$}.  This is true, for example, of a discrete Gaussian distribution on two-power cyclotomics, or a distribution formed by choosing coefficients uniformly from some subset of $\Fq$.  For the former observation, the relevant fact is the following:  \emph{the power basis associated to $\zeta_m$ is orthonormal (after scaling) in the canonical embedding.}  To see this, use \eqref{eqn:trace-pairing} and observe that if $\zeta_m^a$ has order $2^\ell \ge 2$, then $-\overline{\zeta_m}^a$ does also, hence the real parts of the complex embeddings of roots of unity form a collection symmetrical about zero.  For this paper, we will concern ourselves exclusively with this case.
 
 \subsection{Secret distribution}

 We will not make any particular assumption on the secret distribution.  It may be taken to be uniform on $R_q$.  Note, however, that the method of \cite[Section 3, Targeting $e_i$]{Bunny} could be used to manipulate the samples so the secret can be taken from the error distribution, preserving the Ring-LWE structure of the samples.

 \section{Key theoretical properties}
 \label{sec:advantages}

 In this section we highlight several key aspects of Ring-LWE absent in LWE.

 \subsection{Ring homomorphisms}

  If a Ring-LWE problem is presented in $R_q$, then for any $\mathfrak{a} \mid qR$, we have a ring homomorphism
  \[
    \rho:  R_q \rightarrow R/\mathfrak{a}.
  \]
  This transports samples distributed according to $A_{s,\chi}$ to samples distributed according to $A_{\rho(s), \rho(\chi)}$.

  In general, the effect of $\rho$ on $\chi$ is problematic, i.e.\ it spreads out the error widely.  As an illustration, we give a proposition governing the behaviour of $\rho$ on $\chi$ in the $2$-power cyclotomic case, when $q \equiv 1 \pmod 4$.

  \begin{proposition}
    \label{prop:transport}
    Suppose we are in the $2$-power cyclotomic case, and $R/\mathfrak{a} \cong \Fqk$, and $q \equiv 1 \pmod 4$.  
  If, in $R_q$, the error distribution $\chi$ is formed on the $\zeta$-basis in $R_q$ with coefficients drawn from $\chi_0$ on $\Fq$, then $\chi' := \rho(\chi)$ is formed on the $\zeta$-basis in $\Fqk$ with coefficients drawn from $\chi_0'$ on $\Fq$, where $\rho(\zeta_m^k) \in \Fq$ and
  \[
  \chi_0' = \sum_{i=0}^{n/k-1} \rho(\zeta_m^{k})^i \chi_0.
  \]
\end{proposition}

\begin{proof}
  Define $r = \operatorname{ord}_2(q-1)$, meaning that $2^r \mid q-1$ but $2^{r+1} \nmid q-1$.  Since $q \equiv 1 \pmod 4$, we have $r \ge 2$.  Furthermore, $q^i + 1 \equiv 2 \pmod 4$ for all $i$, so that $\operatorname{ord}_2(q^2-1) = \operatorname{ord}_2((q-1)(q+1)) = r+1$ and, by induction
  \[
    \operatorname{ord}_2(q^{2^i}-1) = r+i
  \]
  for all $i \ge 1$.  As $k$ is defined as the embedding degree of the $2n$-th roots of unity, we obtain $k = \frac{2n}{2^r}$.

  The element $\rho(\zeta_m^k)$ satisfies $\rho(\zeta_m^k)^{2n/k} = 1$ in $R/\mathfrak{a}$.  Hence it is itself a primitive $2n/k$-th root of unity, i.e.\ $2^r$-th root of unity.  Hence $\rho(\zeta_m^k) \in \Fq$ by the definition of $r$.

  The main statement now follows from the fact that $1, \zeta_m, \ldots, \zeta_m^{k-1}$ is an $\Fq$-basis of $\Fqk$, that $\rho(\zeta_m^k) \in \Fq$ and that for $0 \le j < k$ and $0 \le i < n/k$, we have
  \[
    \rho(\zeta_m^{ik+j}) 
    = \rho(\zeta_m^{k})^i\rho(\zeta_m^j) 
    = \rho(\zeta_m^{k})^i \zeta_m^j.
  \]
\end{proof}

For example, in the case that $k = n/2$, we obtain
\[
  \chi_0' = \chi_0 + \rho(\zeta_m^k) \chi_0.
\]
This means the coefficients of $\chi'$ are chosen from a sum of two Gaussian distributions with different coefficients.  This is less controlled than twice a single Gaussian.  For, twice a Gaussian is simply a wider Gaussian, and the size of its support grows by approximately $\sqrt{2}$.  However, in an uneven linear combination the size of the support $E_{\chi'}$ is approximately the square of the size of $E_\chi$.  (To be explicit, since $\chi_0$ is discrete, $c\chi_0$ is ``spaced out'' into isolated spikes, and each spike of support is transformed into a small gaussian by the addition of $\chi_0$ to form $c \chi_0 + \chi_0$.)  This is a symptom of the protective property of these ring homomorphisms:  they transform the error to something less amenable to attack.  In fact, very quickly the image of a Gaussian error approaches uniform in the image ring as the dimension of the image ring decreases.  And Ring-LWE samples with uniform error are informationless.

 \subsection{Rotating samples}

 The ring structure allows us to generate new (but not independent) samples from old.  

 \begin{proposition}
   Suppose $\chi$ is invariant under multiplication by $\zeta$.  Then if $(a,b)$ is distributed according to $A_{s,\chi}$, then
   \begin{enumerate}
 \item $(\zeta a, \zeta b)$ is also distributed according to $A_{s,\chi}$,
 \item $(a, \zeta b)$ is distributed according to $A_{\zeta s, \chi}$.
   \end{enumerate}
 \end{proposition}

 In particular, in the $2$-power cyclotomic case, a discrete Gaussian is invariant under multiplication by $\zeta_m$ and all its powers.

 We call these \emph{rotated samples}.  One could also rotate by other small values, e.g.\ $1 + \zeta_m$ in the $2$-power cyclotomic case, at a small cost in changing the error distribution.  (This may allow for adapting the notion of sample amplification to the Ring-LWE case; see \cite{1222}.)

 \subsection{Subrings and trace maps}
 \label{sec:trace}

 If considering Ring-LWE in $R_q$, where $R$ is the ring of integers of a number field $K$, then any subfield $L \subseteq K$ gives rise to a subring $S \subseteq R$ (i.e., the ring of integers of $L$) and, modulo $q$, to a subring $S_q \subseteq R_q$.  Then $S_q$ is an $\Fq$-vector subspace of $R_q$, and $R_q$ has a module structure over $S_q$.  The dimensions of $K$ over $L$, $R$ over $S$ and $R_q$ over $S_q$ agree. 

 There is a linear map $T := \operatorname{Tr}^{R_q}_{S_q} : R_q \rightarrow S_q$ satisfying the following relationship to the usual trace map from $R$ to $S$:
 \begin{equation*}
   \operatorname{Tr}^R_S(x) \bmod{qS} = \operatorname{Tr}^{R_q}_{S_q} ( x \bmod{qR} ).
 \end{equation*}
 To see this, remark that $qS$ is elementwise fixed by the Galois group of $K/L$ and $qR$ is the extension of $qS$ to $R$, so the Galois group takes $qR$ to itself.  Therefore the Galois group acts on $\Rq$ fixing $\Sq$.  Therefore we may define $\operatorname{Tr}^{R_q}_{S_q}(x)$ to be the sum of $\sigma(x)$ for $\sigma$ in the Galois group of $K/L$, and the relationship above holds.

The ring $R$ is always an $S$-module, but the reader is cautioned that in a general number field, $R$ may not be a \emph{free} module over $S$.

 \subsection{Multiplicative cosets of subrings}
 \label{sec:mulcoset}
  
 The set $a_0 S_q$, for any invertible $a_0 \in R_q$, is an $\Fq$-vector subspace of $R_q$ of dimension equal to the dimension of $S_q$.  Distinct such subspaces intersect only at subspaces consisting of non-invertible elements of $R_q$, and $R_q^*$ (the invertible elements of $R_q$) lie in the union of all such subspaces.

 Let us write $A_{a_0\Sq,s,\chi}$ for the distribution on $a_0 \Sq \times \Rq$ given by choosing $a$ uniformly in $a_0\Sq$, choosing $e$ according to error distribution $\chi$ and outputting $(a, b:=as+e)$.

 \begin{proposition}
   If $(a,b)$ is distributed according to $A_{a_0\Sq,s,\chi}$ where $\chi$ is invariant under multiplication by $\zeta$, then 
   \begin{enumerate}
     \item $(\zeta a, \zeta b)$ is distributed according to $A_{\zeta a_0\Sq, s, \chi}$, and
     \item $(a, \zeta b)$ is distributed according to $A_{a_0\Sq, \zeta s, \chi}$. 
   \end{enumerate}
 \end{proposition}
 
 The multiplicative coset structure gives rise to another type of sample reduction, beyond ring homomorphism.  We have
 \begin{proposition}
   \label{prop:new-trace-hom}
   Suppose $s \in \Rq$ is fixed.  Define $T := Tr^{\Rq}_{\Sq}$, the trace map described above. 
 Consider a collection of samples distributed according to $A_{a_0\Sq,s,\chi}$, where $a_0 \in \Rq^*$ is fixed and $T(a_0)$ is invertible.
 Then $T$ maps such samples to samples distributed according to $A_{s',T(\chi)}$ in $\Sq$, where
 \[
   s' = 
     \frac{T(a_0 s)}{ T(a_0)} .
   \]
\end{proposition}

 \begin{proof}
 For $a = a_0 a' \in a_0\Sq$, since $T$ is $\Sq$-linear, we have
 \[
   T(as) = a' T(a_0 s).
 \]
 This implies that 
 \[
   ( T(a), T(as+e) ) 
   =  \left( a' T(a_0), a'  T(a_0) \left( \frac{T(a_0 s)}{ T(a_0)} \right) + T(e) \right)
 \]
 This proves the proposition.
 \end{proof}

 \subsection{Trace maps for two-power cyclotomics}
 \label{sec:ortho}

 The final piece to the puzzle is the behaviour of the trace map $T$ in the previous section.  In the case of the $2$-power cyclotomics, the trace map is particularly well-behaved in terms of its effect on the error distribution.  In fact, it takes very many of the basis elements $\zeta_m$ to zero.  This is a feature of the orthogonality of the basis $1, \zeta_m, \ldots, \zeta_m^{n-1}$, and it may be proved with reference to basic algebraic number theory, as follows.

 Using the notation of Section \ref{sec:trace} in the case of the $2$-power $m$-th cyclotomics $K$, let $L$ be the $k$-th cyclotomic subfield.
 One may take $\zeta_k = \zeta_m^{m/k}$ and $S_q$ has a basis $1, \zeta_k, \ldots, \zeta_k^{k/2-1}$ over $\FF_q$.  We collect terms to write
 \begin{align*}
   R_q &= \ZZ + \zeta_m \ZZ + \cdots + \zeta_m^{m/2-1}\ZZ \\
   &= (\ZZ + \zeta_k\ZZ + \cdots + \zeta_k^{k/2-1}\ZZ) + \zeta_m( \ZZ + \zeta_k\ZZ + \cdots + \zeta_k^{k/2-1}\ZZ) \\
   & \quad\quad\quad\quad\quad + \cdots + \zeta_m^{m/k-1}( \ZZ + \zeta_k\ZZ + \cdots + \zeta_k^{k/2-1}\ZZ) \\
   &= \Sq + \zeta_m \Sq + \cdots + \zeta_m^{m/k-1} \Sq.
 \end{align*}
 In other words, $R_q$ has a $\zeta$-basis over $S_q$. 

 The elements of the Galois group of $K/L$ are given by $\zeta_m \mapsto \zeta_m^{a}$ for $a \in (\ZZ/m\ZZ)^*$ satisfying $a \equiv 1 \pmod k$, and so
 \begin{align*}
   \operatorname{Tr}^{R_q}_{S_q}(\zeta_m^i) &= \sum_{\substack{0 \le a < m \\ a \equiv 1 \pmod k}} \zeta_m^{ia} \\
   &= \zeta_m^i \sum_{a=0}^{m/k-1} \zeta_m^{iak} \\
   &= \left\{
     \begin{array}{ll}
       0 & i \not\equiv 0 \pmod{\frac{m}{k}} \\
       \frac{m}{k}\zeta_m^i & i \equiv 0 \pmod{\frac{m}{k}} \\
     \end{array}
     \right. .
 \end{align*}

 In particular, for the trace to the index two subfield, we have:
\begin{equation*}
     \frac{1}{2} T^{R_q}_{S_q}(\zeta_m^i) = \left\{
       \begin{array}{ll}
	 0 & i \equiv 1 \pmod 2 \\
	 \zeta_m^i & i \equiv 0 \pmod 2
       \end{array} \right. .
     \end{equation*}
     This special case can be seen directly by observing that if $i$ is even, then $\zeta_m^i \in S$, while if $i$ is odd, then $\zeta_m^i$ is the square root of something in $S$, i.e.\ it satisfies the minimal polynomial $x^2 - \zeta_m^{2i}$, and hence has trace zero.  An alternate proof of the general case then follows by application of the special case $\log_2(m/k)$ times.

 In summary then, the trace map preserves the error distribution up to small factors.  The following proposition, which is now immediate, makes this explicit.

 \begin{proposition}
   \label{prop:new-error-trace}
   Suppose we are in the two-power cyclotomic case as in Section \ref{sec:notation}, where in particular $R$ is the ring of integers of the $m$-th cyclotomics, with $m$ a power of two.  Let $S$ be the subring of integers of the $k$-th cyclotomics (hence $k$ is also a power of two).   Write $T := Tr^{\Rq}_{\Sq}$ for the trace map described in Section \ref{sec:trace}.  Suppose that $\chi$ is an error distribution formed on the $\zeta$-basis of $\Rq$ with coefficients chosen according to $\chi_0$.  Then $\frac{k}{m}T$ takes values in $\Sq$ and $\frac{k}{m}T(\chi)$ is the error distribution formed on the $\zeta$-basis of $\Sq$ with coefficients from $\chi_0$.
 \end{proposition}

 The efficacy of the trace map with respect to the error distribution is due to its being an orthogonal projection to the space spanned by a subset of an orthonormal basis.

 \section{Reducing to a smaller ring}
 \label{sec:reduction}

 We demonstrate that if one can find sufficiently many samples whose $a$ values are restricted to a fixed multiplicative coset of a subring, then we can reduce the Ring-LWE problem to multiple independent Ring-LWE instances in the subring, without error inflation.

 For this section, we are in the two-power cyclotomic case.  Let $R$ be the ring of $m$-th cyclotomic integers, where $m$ is a power of two (which have dimension $n$, where $m = 2n$), and $S$ be the ring of $k$-th cyclotomic integers, where $k \mid m$.  Then we have an extension of rings, $S \subseteq R$ of degree $m/k$.  Suppose that the rational prime $q$ is unramified in $R$.

 \begin{proposition}
   \label{prop:reduction}
   Consider a Ring-LWE instance in $R_q$ with secret $s$ and error distribution $\chi$.  Let $a_0 \in R_q$ be a fixed invertible element.  
   Let $T := \operatorname{Tr}^{\Rq}_{\Sq}$, and suppose that $T(a_0)$ is invertible.

   Let $i$ be an integer. Then in time polynomial in $n$ and $\log q$, one can reduce a Ring-LWE sample from distribution $A_{a_0\Sq,s,\chi}$ to a Ring-LWE sample in $\Sq$ drawn according to secret
   \[
     \frac{T(a_0 \zeta^i s)}{T(a_0)}
   \]
   and error distribution $\frac{k}{m}T(\zeta^i \chi) \subseteq \Sq$.
 \end{proposition}

 In particular, by Proposition \ref{prop:new-error-trace}, coefficient distributions of a $\zeta$-invariant $\chi$ and its resulting distribution $\frac{k}{m}T(\zeta^i \chi)$ are of the same size; it is in this sense that the errors do not inflate.

 \begin{proof}
   Consider the sample $(a,b)$ where $b= as+e$.  Multiplying the second coordinate of the sample by $\zeta^{i}$ and taking the trace $\frac{k}{m}T$, we obtain as in Proposition \ref{prop:new-trace-hom}, a sample
   \begin{align*}
     &\left( \frac{k}{m}T(a), \frac{k}{m}T(\zeta^i b) \right)  \\
     &\left( \frac{k}{m}T(a), \frac{k}{m}T(a\zeta^{i}s+\zeta^{i}e) \right)  \\
     &= \left(a'\frac{k}{m}T(a_0), a'\frac{k}{m}T(a_0)\cdot \left( \frac{T(a_0\zeta^{i}s)}{T(a_0)} \right) +\frac{k}{m}T(\zeta^{i}e) \right),
   \end{align*}
   where $a' := aa_0^{-1} \in S_q$.  
   
   Multiplication in the ring, and taking the trace, are polynomial in the ring size.
 \end{proof}

 The following is the main theorem of the paper.

 \begin{theorem}
   \label{thm:reduction}
Suppose $R$ is the ring of $m$-th cyclotomic integers, for $m = 2n$ a power of two, and $S$ is the ring of $k$-th cyclotomic integers, where $k \mid m$, so the extension $S \subseteq R$ is of degree $m/k$.  Suppose that the rational prime $q$ is unramified in $R$.

   Consider a Ring-LWE instance in $R_q$ with secret $s$ and error distribution $\chi$ which is invariant under multiplication by $\zeta = \zeta_m$, a primitive $m$-th root of unity.  Let $a_0 \in R_q$ be a fixed invertible element.  
   Let $T := \operatorname{Tr}^{\Rq}_{\Sq}$ as defined in Section \ref{sec:trace}, and suppose that $T(a_0)$ is invertible.

   Suppose one obtains $N$ samples $(a,b)$ distributed according to $A_{a_0\Sq,s,\chi}$ (notation from Section \ref{sec:mulcoset}).  

   Then in time linear in the number of samples $N$, and polynomial in $n$ and $\log q$, one can reduce the computation of the secret $s \in \Rq$ to the solution of $m/k$ Search Ring-LWE problems in $\Sq$ with error distribution $\frac{k}{m}T(\chi)$, having $N$ samples each.  These $m/k$ problems are independent in the sense that setting up any one of them does not require having solved any other one.

   Furthermore, if $\chi$ is formed on the $\zeta$-basis from coefficient distribution $\chi_0$ on $\Fq$ (see Section \ref{sec:error} for definition), then so is $\frac{k}{m}T(\chi)$.
 \end{theorem}

 \begin{proof}
   Set $i = j$ in Proposition \ref{prop:reduction} for each $j$ in the range $j=0,\ldots,m/k-1$, to obtain $N$ samples having secret
   \[
     c_j:=\frac{T(a_0 \zeta^{j}s)}{T(a_0)}.
   \]
   Using an oracle that solves Search Ring-LWE in $\Sq$, obtain $c_j$.

   Collecting all the values $c_j$, we have a linear system of $m/k$ equations over $\Sq$, whose indeterminates are the coefficients of $s$ (expressed in terms of a basis for $\Rq$ over $\Sq$), of the form
   \[
     T(a_0 \zeta^{j}s) = c_{j}T(a_0), \quad j=0, \ldots, m/k-1.
   \]
   The linear equations are independent provided that $\{ a_0\zeta^{j} \}$ is a set of $\Sq$-independent vectors in $\Rq$.  We saw above that $\{ \zeta^{j} \}_{j=0,\ldots,m/k-1}$ is a basis for $\Rq$ over $\Sq$.  Thus independence is guaranteed by the fact that $a_0$ is invertible.   Note that we can consider this system to consist of $n$ independent linear equations over $\Fq$. The system can be solved by Gaussian elimination to recover $s$.

  All the field operations concerned are polynomial in the size of the ring.  We must apply the trace to $N$ samples $m/k$ times, and we must carry out Gaussian elimination of dimension $n = m/2$ over $\Fq$, which is polynomial in $m$ and $\log q$.
 \end{proof}

 As a small corollary, note that in any small Ring-LWE situation where exhaustive search may apply, it is equally possible to use the above for a square-root speedup, provided many samples are available.  As an example, if we have a coefficient distribution with support not including all of $\Fq$, then the following statement demonstrates the approach.

 \begin{corollary}
   \label{corollary:square-root}
   Consider a Ring-LWE instance in $\Rq$ with secret $s$ and error distribution $\chi$ formed on a $\zeta$-basis with coefficient distribution having support strictly smaller than $\Fq$.

   There is an algorithm to solve this problem, with success probability $1/2$, in time and number of samples $q^{n/2}$ times factors polynomial in $n\log q$, using space polynomial in $n\log q$.
 \end{corollary}

 \begin{proof}
   Note that the hypotheses guarantee $\chi$ is invariant under multiplication by $\zeta$.  Let $\Sq$ be the ring of index two in $\Rq$ (i.e.\ $n$-th roots of unity).
   Collect samples, discarding all but those with $a \in \Sq$.  In time $O(N q^{n/2})$ we can accumulate $N$ samples with $a \in \Sq$.  Apply Theorem \ref{thm:reduction} to reduce to two Ring-LWE problems in $\Sq$ with $N$ samples each.  The error distribution $\chi$ on $\Rq$ gives an error distribution $\frac{1}{2}T^{\Rq}_{\Sq}(\chi)$ on $\Sq$.  If $\chi$ is formed on a $\zeta$-basis with coefficients supported in $E_\chi \subsetneq \Fq$, then $\frac{1}{2}T^{\Rq}_{\Sq}(\chi)$ is formed on a $\zeta$-basis with coefficients supported in $E_\chi \subsetneq \Fq$.  Therefore, if the number of samples is sufficient, the reduced Ring-LWE problems are solvable using exhaustive search through possible $s$ values.

   In our case, we need $N$ large enough so that a Ring-LWE problem in $\Sq$ with $N$ samples has a unique solutions with probability $1/\sqrt{2}$.  Although $N$ depends upon $|E_\chi|$, for the worst case $|E_\chi| = q-1$, $N$ is still polynomial in $n\log q$.  Solve the reduced problems by exhaustive search, which takes time $O(q^{n/2})$ and each succeeds with probability $1/\sqrt{2}$.
 \end{proof}

  \section{Background on the Blum-Kalai-Wasserman algorithm}
  \label{sec:bkwback}

  First, we will give a very brief overview of the Blum-Kalai-Wasserman (BKW) algorithm in the context of LWE \cite{bkw}.  It is a combinatorial algorithm in which samples are collected and stored so as to facilitate the creation of new samples, as iterated sums and differences of established ones.  The goal is to create new samples for which $a$ is restricted to a linear subspace.  This is the \emph{reduction phase} of the full BKW algorithm.  
  
  In BKW, after reduction, there is a \emph{hypothesis testing phase}, in which one solves a lower-dimensional Ring-LWE problem (that given by restricting $a$ to the subspace) by exhaustive search over possible secrets.  And then there is a \emph{back-substitution phase}, where the small piece of the secret recovered in hypothesis testing is used to rework the problem to prepare the next small piece for hypothesis testing.

  One can think of BKW as a sort of controlled Gaussian elimination on a matrix whose rows are samples, in which one wants to obtain as much simplification as possible using just one sum or difference of rows.  By keeping the coefficients of the linear combinations small, we prevent the error `blow-up' that occurs with regular Gaussian elimination.  The cost is in needing many more matrix rows (samples) in order to be able to choose good linear combinations.  The back-substitution phase is analogous to the eponymous phase of Gaussian elimination, with the recovered portion of the secret taking the role of the free variable.  From another point of view, BKW reduction is a sort of iterated birthday attack, in which one searches for and exploits collisions which eliminate entries of the vectors, reducing to a subspace, where one searches again for collisions, and so on.  
  
  Now let us be more precise.  During the reduction phase, only the $a$-value of a sample matters, considered as a vector in a vector space $V$, and the goal is to create samples with $a \in W$, a linear subspace of $V$.  Suppose, for the sake of explanation, that $W$ is defined by the first $r$ coefficients of its vectors being $0$.  One generates an ordered list of the first $r$ entries of all the vectors $a$ which are observed.  Whenever a new vector $a$ is observed, it is compared to the ordered list.  If it is not already present, it is added.  Otherwise, we have discovered two samples $(a,b)$ and $(a',b')$ for which $(a-a',b-b')$ is a new sample for which $a-a'$ lies in $W$.  The penalty is that the error distribution of these new samples is widened.  We begin a new table of such vectors as they are generated.  In this way, we produce a large number of samples in a smaller subspace at the cost of inflating the error widths.

  Instead of performing this reduction all at once, one chooses an appropriate \emph{block size} $\beta$ for BKW (which is fixed throughout in the na\"ive implementation), which is to say, the codimension of $W$ as a subspace of $V$.  Once we have produced enough samples in $W$, we can use these to perform another BKW reduction to a subspace $W' \subseteq W$ of codimension $\beta$ in $W$.  The cost of a reduction step is exponential in $\beta$, so we keep $\beta$ as small as possible.  We perform block reductions until the samples are all taken from a small enough subspace to run an exhaustive search or other strategy to finish off the problem.  The limiting factor on shrinking $\beta$ is an upper limit on the number of blocks used overall.  Each reduction into codimension $\beta$ has a cost in error-inflation.  We have a limit on the total error inflation (because hypothesis testing will fail if the error is so inflated as to appear uniform), which limits the total number of blocks.

  The BKW algorithm has been improved in recent years, including using coding theory to reduce the number of values that need to be stored and compared, sieving at each step, allowing the block size to vary, using the Fourier transform to speed up hypothesis testing; see \cite{albrecht-bkw} \cite{better} \cite{sieve-bkw} \cite{asympt} \cite{coded-bkw} \cite{improved-bkw}.

  \section{Reduction using BKW}
  \label{sec:reduction-phase}

  In this section, we address the problem of finding sufficiently many samples $(a,b)$ having $a$ from an appropriate subring $\Sq \subseteq \Rq$, so that Theorem \ref{thm:reduction} will apply.  For this, we use the reduction phase of the BKW algorithm.  We emphasize that it is possible, once the samples have been given in an appropriate basis, to use an off-the-shelf BKW reduction algorithm, including coded BKW with sieving etc., for the reduction phase.  The window size may be chosen at will, for example, and need not depend upon the ring structure.  Then, Theorem \ref{thm:reduction}, which is polynomial time, replaces all the other phases of BKW.

  The only adaptor necessary to connect BKW to Theorem \ref{thm:reduction} is an attention to the basis used.  In order to perform the reduction, we begin with the $\zeta$-basis of $\Rq$ over $\Fq$, namely 
 \[
   1, \zeta, \zeta^2, \ldots, \zeta^{n-1},
 \]
 and then reorder it to produce a \emph{prioritized basis}. The most important property we desire for our purposes is that if one of $\zeta^i$ and $\zeta^j$ has lower multiplicative order than the other, then it comes later than the other.
 One computationally convenient way to accomplish this is to take the bit-reversal permutation on $n$ elements (i.e.\ $a$ maps to $b$ if the binary representation of $a$ in $\log_2(n)$ bits, read backwards, is $b$), then reserve the order.
 For concreteness, the prioritized basis (in part) is as follows:
 \[
   \zeta_m^{n-1}, 
   \zeta_m^{\frac{n}{2}-1},
   \zeta_m^{\frac{3n}{4}-1},
   \zeta_m^{\frac{n}{4}-1},
   \ldots,
   \zeta_m^{\frac{3n}{4}},
   \zeta_m^{\frac{n}{4}},
   \zeta_m^{\frac{n}{2}},
   1.
 \]

 Using any type of BKW reduction, one now reduces, with respect to this basis.  To be precise, one seeks to eliminate the earlier coefficients of the elements $a$, as expressed in this basis.  At the end, at most the last $2^k$ coefficients are non-zero, for some small $k$.  For example, one may reduce until only the last $1$, $2$, $4$ or $8$ coefficients are possibly non-zero.  
 
 The varying block sizes during the reduction algorithm itself need not respect any restrictions, and improvements such as coded-BKW with sieving, may be used.  For example, coded-BKW, under the assumption the secret $s$ is small, associates to each $a$ a codeword $c$ from a linear code.  Then the sample $(a,as+e)$ is replaced with $(c,as+e)$, which is a valid sample with a larger error, before it is fed to the BKW tables.  The tables then have fewer rows because their rows are chosen from codewords.  In sieving, imagine that one has stored the original $a$ along with each new sample $(c,as+e)$.  The difference between $a$ and $c$ measures the error inflation introduced by coding.  A collision between $(c,a_1s+e)$ and $(c,a_2s+e)$ being passed to another table has an $a = a_1-a_2$ that is not actually $0$ in the first few entries, only small.  Among the vectors being fed from one table to the next, one can pause to sieve them, creating vectors whose $a$'s are somewhat smaller.  This reduces the error inflation introduced by the coding process.

The important thing is that, whatever technique is used, after reduction, one has obtained samples with $a \in \Sq = S/qS$ for some $S$ of dimension $2^k$.  One then applies Theorem \ref{thm:reduction}.

 \section{The Ring-BKW algorithm}
 \label{sec:ring-bkw-alg}

 In this section we summarize the Ring-BKW algorithm for completeness.  In short, one uses an off-the-shelf BKW reduction algorithm on samples with respect to a particular choice of basis, then applies Theorem \ref{thm:reduction}.  The important point is that the back-substitution phase of BKW is no longer needed, and the hypothesis-testing phase can be parallelized.  The hypothesis-testing phase can also be off-the-shelf, including recent improvements using the Fourier transform etc. \cite{better}.  However, we will elaborate somewhat.

\subsection*{Ring-BKW algorithm}

Choose a subring $S \subseteq R$ of dimension $B$ over $\ZZ$ (corresponding to a lower-degree $2$-power cyclotomic field), to which we wish to reduce.  Define $\Rq$ and $\Sq$ as before.  The Ring-BKW Algorithm is given as Algorithm \ref{alg}.

\begin{algorithm}
  \caption{Ring-BKW Algorithm}
  \label{alg}
 \begin{enumerate}
   \item Run BKW Reduction
   (as in Section \ref{sec:reduction-phase} above with prioritized basis) on the values $a$ until all samples $(a,b)$ have $a \in \Sq$.  
   \item Use Theorem \ref{thm:reduction} to create samples from $n/B$ different Ring-LWE problems in $\Sq$.  
   \item Solve these Ring-LWE problems using any method of choice.
   \item Use Theorem \ref{thm:reduction} to recover the secret $s$ in polynomial time from these solutions.
 \end{enumerate}
\end{algorithm}

   The ring structure is not relevant in step (1); one uses BKW reduction as for any LWE problem (in particular, the window size can be chosen without regard to the ring structure).  In fact, any reduction algorithm to obtain values $a \in \Sq$ will do as well.  
   
   The following theorem relates any reduction algorithm to the solution of Search Ring-LWE.  For the following, we consider Gaussian error with a well-defined width; an \emph{expansion factor} refers to a multiplicative factor on the width.

 \begin{theorem}
   Suppose that $\mathcal{B}$ is an algorithm which, given a Ring-LWE problem of dimension $n$ over $\Fq$, produces $N$ Ring-LWE samples of dimension $B$ with error expansion factor of $f$, in time $t_\mathcal{B}(n,B,f,N)$, and using $r_\mathcal{B}(n,B,f,N)$ original samples.
   
   Suppose that $\mathcal{R}$ is an algorithm which solves Ring-LWE in dimension $B$ over $\Fq$ in time $t_\mathcal{R}(B)$, given error width less than or equal to $w$ and at least $N$ samples.

   Then, there is an algorithm $\mathcal{A}$ which solves Ring-LWE in $\Rq$ having width $\sigma$ in time 
   \[
     t_\mathcal{B}(n,B,w/\sigma,N) + \frac{n}{B}t_\mathcal{R}(B) + N \cdot (\mbox{time polynomial in $n \log q$}),
   \]
   using $r_\mathcal{B}(n,B,w/\sigma,N)$ samples.
 \end{theorem}

 \begin{proof}
  We will use Algorithm \ref{alg}.  
  We will set $f = w/\sigma$.
  The time to run the reduction phase is $t_\mathcal{B}(n,B,w/\sigma,N)$.  The time to create the smaller Ring-LWE problems is linear in $N$ and polynomial in $n\log q$ from Theorem \ref{thm:reduction}.  Solving the $\frac{n}{B}$ smaller Ring-LWE problems (guaranteed to succeed by the choice of $f$) takes time $t_\mathcal{R}(B)$ each.  Then reconstructing the secret (as in Theorem \ref{thm:reduction}) again takes polynomial time.
\end{proof}

 \section{Advanced Keying}
 \label{sec:advanced-keying}

 In the previous section, one uses BKW on LWE to perform reduction, say with block size $B$.  Given a Ring-LWE sample, there are in fact $n$ rotated samples one could feed into the reduction:
 \[
   (a,b), (\zeta a, \zeta b), \ldots, (\zeta^{n-1} a, \zeta^{n-1} b).
 \]
 Na\"ively, one may include them all, or include the first one.  Probably the best course of action is to include them all, to increase the number of collisions located amongst the available samples (since the number of samples needed is the downside to BKW in general).  By including all rotations, one catches all collisions of the form $a_1 \pm \zeta^i a_2$ for some $i$.  These are all perfectly useful collisions for the algorithm, if the error term is $\zeta$-invariant.  In this section we propose a space-saving approach based on symmetries, which is equivalent, in terms of collisions obtained per sample, to storing all rotations of the samples.  (If one chooses to compare to running BKW without rotating samples at all, i.e.\ ring-blind, it will both reduce storage \emph{and} require fewer samples.)

 In the discussion that follows, the reduction algorithm described in Section \ref{sec:bkwback} will be called \emph{traditional BKW reduction} to distinguish it from the \emph{advanced keying BKW reduction} proposed in this section.  There are a variety of modern speedups and alternatives (such as coded-BKW and sieving) which could also be combined with advanced keying, but for purposes of clarity we will ignore these until later in this section.  In particular, in traditional BKW reduction, when a collision is recorded, nothing is added to the current table, but the difference is passed to the next table.  (Later, it will prove helpful to call this \emph{one-difference} and compare it to \emph{all-differences} where new samples are stored as well as passed on, to increase the number of collisions.)

Our proposal in this section is an analogue of the space-saving technique used in traditional BKW, wherein for each sample $(a,b)$ we may derive two samples $(a,b)$ and $(-a,-b)$:  we choose one canonically (where the first non-zero coefficient of $a$ is in $\left\{1,\ldots,\frac{q-1}{2}\right\}$, say), and save only this one.  By doing so, we will catch all collisions between samples where their sum \emph{or their difference} vanishes, and save half the table rows in the process.  More precisely, the number of rows of the table for each block never exceeds $(q^B-1)/2$, since the possible non-zero vectors come in pairs of which we store at most one.
 Furthermore, this is also a time efficiency issue.  If instead one simply included $(a,b)$ and $(-a,-b)$ among the incoming samples, then without this trick, the collisions $a_1 + a_2$ and $-a_1 - a_2$ are both sent on to the next table, both are multiplied by $-1$ thereafter, and we actually end up with repeat samples that must be weeded out at a later stage.  For reference, traditional BKW reduction, with this space-saving technique, is given explicitly in Algorithm \ref{alg2}.

The fundamental observation is that the prioritized basis proposed in the last section is particularly well-suited to this type of strategy, because of the resulting `negacyclic permutation' effect of multiplication by $\zeta$.  It results in a savings of $1/2B$ instead of $1/2$ and is completely analogous to the trick above in both space and efficiency savings.  It requires that the block size $B$ be a power of $2$.

Write $\mathbf{a} \in \Fq^n$ for the vector of coefficients of $a$ in the prioritized basis.  
The action of $\zeta^h$ (taking $a$ to $\zeta^h a$) on such a vector permutes the entries, and swaps the sign on some of them (since $\zeta^n = -1$).  Suppose $h$ is exactly divisible by $2^\ell$ (i.e. $\ord_2(h) = \ell$).  With regards to the permutation only (ignoring the signs), the permutation has the property that it stabilizes each consecutive block of length $n/2^\ell$ throughout (that is, it permutes each block individually).  For fixed $\ell$, there are exactly $n/2^\ell$ such integers $h$ (note that $h$ is taken modulo $2^n$, for $h=2^n$ results in the identity permutation).  The following consequence is key:

\begin{property}
  \label{prop:zero}
  Let $B \mid n$ denote block size.  Then applying $\zeta^{n/B}$ preserves the property that $\mathbf{a}$ has first block (or series of any number of first blocks) consisting of zero entries.
\end{property}

This property will allow us to rotate samples by any of the $B$ quantities $1, \zeta^{n/B}, \zeta^{2n/B}, \ldots, \zeta^{(B-1)n/B}$ during BKW reduction with block size $B$. 

Next, one must specify a \emph{canonical choice of representative} from the set of possible rotations $\{a, \zeta^{n/B}a, \ldots, \zeta^{(B-1)n/B}a\}$, depending only on the first non-zero block of entries, up to an overall sign.  A possible canonical choice is the ordering which has smallest first entry (in absolute value), together with some tie-breaking conventions, e.g.\ smallest second entry, etc., and if all entries are equal in absolute value, then some appropriate convention on sign changes between $\mathbf{a}$ and $|\mathbf{a}|$, etc.  However, any ordering of the possible length-$B$ vectors modulo overall sign, will do.  It is not possible to break a tie if the first $B$ entries of the two rotations actually \emph{agree} up to overall sign under one of the rotations.  However, in this case we have found a ``self-match,'' meaning that two of the rotations have a difference which has all zero in the block under consideration, and so at most one of the two rotations need be stored, and the difference is sent to the following block, as with any collision, as in a traditional BKW algorithm.

The advanced keying BKW reduction is given in Algorithm \ref{alg2}, and for comparison purposes, the traditional BKW reduction using all rotations of each sample is given in Algorithm \ref{alg3}.

\begin{algorithm}
  \caption{Traditional BKW Reduction Phase}
  \label{alg3}
  \begin{algorithmic}[1]
    \STATE Create empty Tables $1$ through $n/B$.

    \FOR{each initially available sample $(a,b)$}
    	\FOR{$j=0$ \TO $n-1$}
       		\STATE Rotate the sample by $\zeta^j$, to obtain $(a_1,b_1)$.
		\STATE Send the sample $(a_1,b_1)$ to Table $1$.
	\ENDFOR
    \ENDFOR

     \FOR{each sample $(a,b)$ sent to Table $i$, $i < n/B$}
     \IF{$a$ has all $0$ entries in block $i$}
     \STATE{send sample $(a,b)$ on to Table $i+1$}
     \ENDIF
   \STATE Multiply by $-1$ if necessary to ensure the first non-zero coefficient of $a_1$ is in the range $1$ to $(q+1)/2$.
   \IF{a collision is found (i.e.\ a sample $(a_0,b_0)$ already exists in the table having the same first $i$ blocks of size $B$)}
       \STATE Subtract $(a_1,b_1)$ from $(a_0,b_0)$ to obtain a new sample whose first $i$ blocks of size $B$ are zero
       \STATE Send the result to Table $i+1$.
   \ELSE
       \STATE Store the associated sample in Table $i$.
     \ENDIF
       \ENDFOR
     \end{algorithmic}
\end{algorithm}

\begin{algorithm}
  \caption{Advanced Keying BKW Reduction Phase}
  \label{alg2}
  \begin{algorithmic}[1]
    \STATE Create empty Tables $1$ through $n/B$.

    \FOR{each initially available sample $(a,b)$}
    	\FOR{$j=0$ \TO $n/B-1$}
       		\STATE Rotate the sample by $\zeta^j$, to obtain $(a_1,b_1)$.
  		\STATE Send the sample $(a_1,b_1)$ to Table $1$.
   	\ENDFOR
     \ENDFOR

     \FOR{each sample $(a,b)$ sent to Table $i$, $i < n/B$}
     \IF{$a$ has all $0$ entries in block $i$}
     \STATE{send sample $(a,b)$ on to Table $i+1$}
     \ENDIF
   \STATE From $a, \zeta^{n/B}a, \ldots, \zeta^{(B-1)n/B}a$, choose a canonical representative.
   \FOR{every sample $(a_1,b_1)$ corresponding to a canonical representative}
   \STATE Multiply by $-1$ if necessary to ensure the first non-zero coefficient of $a_2$ is in the range $1$ to $(q+1)/2$.
   \IF{a collision is found (i.e.\ a sample $(a_0,b_0)$ already exists in the table having the same first $i$ blocks of size $B$)}
       \STATE Subtract $(a_1,b_1)$ from $(a_0,b_0)$ to obtain a new sample whose first $i$ blocks of size $B$ are zero
       \STATE Send the result to Table $i+1$.
   \ELSE
       \STATE Store the associated sample in Table $i$.
     \ENDIF
       \ENDFOR
       \ENDFOR
     \end{algorithmic}
\end{algorithm}

Correctness of Algorithm \ref{alg2} is a consequence of Property \ref{prop:zero}.  Furthermore, Algorithms \ref{alg3} and \ref{alg2} catch the same collisions in the following heuristic sense.  For each collision $\zeta^i a_1 - \zeta^j a_2$, there will be another collision at $\zeta^{i+k} a_1 - \zeta^{j+k} a_2$ for any $k \equiv 0 \pmod{n/B}$.  In Algorithm \ref{alg3}, all $B$ of these collisions are passed on to the next table after storing $B$ new rows in the current table.  But any one of the samples sent on can generate the others via rotation, so only one of them is actually needed at the next table.  In Algorithm \ref{alg2}, only one of them is stored and only one is sent onward (but only one is needed).  However, there is some difference in the final output because we are only keeping one sample per row, and the order of input samples to a given table may differ, resulting in a different table entry.  If one uses the \emph{all-differences} variation, this difference disappears and the output of the two algorithms will be the same.

The following is immediate from Algorithm \ref{alg2}.

\begin{proposition}
  Each table in Algorithm \ref{alg2} has at most $\frac{q^B-1}{2B}$ rows in total.
\end{proposition}

Finally, we will remark again that BKW reduction improvements for LWE, such as coded-BKW and sieving, may also be adapted to use the advanced keying demonstrated here, \emph{provided} block sizes can be maintained to be powers of $2$ (varying them is ok).  As some modern algorithms vary block size, this may be an impediment.  The na\"{i}ve way to do this would be to code samples first, then choose a canonical rotation of each codeword.  Perhaps better, one could also code each rotation and choose the one with smallest error, which may introduce a significant improvement to the error inflation, depending on the choice of code.  (Note that, for those familiar with coded-BKW, the notion of advanced keying is not so different than coding, as it provides a sort of 'codeword' for each sample, without an error inflation.)

Algorithms \ref{alg3} and \ref{alg2}, as well as a completely ring-blind version of BKW reduction were coded in Python in Sage Mathematics Software for comparison purposes.  Some example results are given in Table \ref{table:adv}.  In short, the advanced keying did reduce table sizes and samples needed as described, and had a faster overall runtime.  A few remarks are in order:

\begin{enumerate}
  \item The experiments were chosen to represent a range of small parameter sets, where timings were in the range of seconds or minutes on a Lenovo X1 laptop.
  \item After parameters were chosen, the number of samples was chosen to be a round number where the final table began to have a few samples on average; the timing therefore roughly represents the time until the final table begins to populate.
  \item To compare meaningfully, the ring-blind algorithm uses $n$ times as many initial samples, which is equal to the total number of rotations of incoming samples for the other algorithms.  The fact that the final table is populated but not full in all cases is evidence that the number of samples needed by Algorithms \ref{alg2} and \ref{alg3} is $1/n$ of those needed na\"{i}vely.
  \item For some smaller parameter sets, we also tested a version of the algorithm (labelled AD = `All Differences') in which every sample encountered is stored (so each row of the table can contain multiple samples) and every difference is passed on (i.e. the new sample is compared to everything already in its row).  The purpose of this is to demonstrate that the advanced keying will still find the same number of samples.  However, the AD version is significantly slower in all cases, so it was only implemented for some of the smaller parameter sets in the table.
  \item Algorithms \ref{alg3} and \ref{alg2} are pseudocode; the implementation necessarily addressed details not covered in the pseudocode presentation.  For example, some moderate attention was given to efficiency in the rotation of samples.  For example, when only certain coefficients of the rotation were needed, only those were computed.
\end{enumerate}

Some experimental observations:

\begin{enumerate}
  \item The table sizes observed in Algorithm \ref{alg2} are very close to $1/B$ of the number observed in Algorithm \ref{alg3}, as expected.
  \item The faster runtime of Algorithm \ref{alg2} is a result of the fact that fewer samples are handled ($1/B$ as many are fed to the first table compared to Algorithm \ref{alg3}), although they must be handled in more detail, so the speedup is less than a $1/B$ factor.
  \item Algorithms \ref{alg3} and \ref{alg2} use the exact same starting data, and it is reassuring that the reduced sample counts are similar, and the same in the AD version. 
  \item Algorithm \ref{alg3} tends to find more samples than Algorithm \ref{alg2}.  The difference is in which matches are found when more than two samples collide in a row in the table, and therefore is more pronounced as the number of rows grows.  
\end{enumerate}

\begin{table}
  \label{table:adv}
  \begin{tabular}{l|l|l|l}
    $n=2^3$, $B=2^2$, $q=211$ & Ring-blind & Algorithm \ref{alg3} & Algorithm \ref{alg2} \\
    \hline
    \hline
    Initial Samples & $4000 \cdot 2^3$ & $4000$ & $4000$ \\
    OD Table Size & $31999$ & $31996$ & $7999$ \\
    OD Reduced Samples & $1$ & $1$ & $1$ \\
    OD Runtime & $1.43$ s & $1.80$ s & $2.42$ s \\
    AD Table Size & $31999$ & $31996$ & $7999$ \\
    AD Reduced Samples & $1$ & $1$ & $1$ \\
    AD Runtime & $1.73$ s & $1.92$ s & $2.46$ s \\
	\hline
	\hline
    $n=2^4$, $B=2^2$, $q=17$ & Ring-blind & Algorithm \ref{alg3} & Algorithm \ref{alg2} \\
    \hline
    \hline
    Initial Samples & $2000 \cdot 2^4$ & $2000$ & $2000$ \\
    OD Table Size & $31985$ & $31988$ & $7997$ \\
    OD Reduced Samples & $15$ & $3$ & $3$ \\
    OD Runtime & $6.13$ s & $6.24$ s & $4.69$ s \\
    AD Table Size & $36448$ & $36623$ & $9163$ \\
    AD Reduced Samples & $81$ & $21$ & $21$ \\
    AD Runtime & $8.43$ s & $9.73$ s & $6.27$ s \\
    \hline
    \hline
    $n=2^5$, $B=2^2$, $q=7$ & Ring-blind & Algorithm \ref{alg3} & Algorithm \ref{alg2} \\
    \hline
    \hline
    Samples & $200\cdot 2^5$ & $200$ & $200$ \\
    OD Table Size & $6368$ & $6386$ & $1596$ \\
    OD Reduced Samples & $31$ & $13$ & $4$ \\
    OD Runtime & $7.39$ s & $8.07$ s & $4.23$ s \\
    \hline
    \hline
    $n=2^6$, $B=2^3$, $q=3$ & Ring-blind & Algorithm \ref{alg3} & Algorithm \ref{alg2} \\
    \hline
    \hline
    Samples & $250\cdot 2^6$ & $250$ & $250$ \\
    OD Table Size & $15988$ & $15993$ & $1998$ \\
    OD Reduced Samples & $12$ & $7$ & $2$ \\
    OD Runtime & $27.0$ s & $29.0$ s & $10.2$ s \\
    \hline

  \end{tabular}
  \caption{\small The term ``Ring-blind'' refers to Algorithm \ref{alg3} but with $j=0$ to $0$ in line 3, i.e. without rotating any initial samples.  A fixed list of samples was generated pseudorandomly for each experiment; `Initial Samples' refers to how many were used from the beginning of the list.  `Reduced Samples' refers to the number of samples eventually contained in the last table.  `Runtime' refers to the wall time as measured in Sage Mathematics Software.  `Table Size' refers to the total number of rows stored not counting the final samples.  `OD' (One Difference) refers to the algorithms as presented in the paper.  `AD' (All Differences) refers to a modification in which every sample that matches a row is also stored in that row, and when a match is found, the differences with everything in the row are passed on.}
\end{table}

\section{In practice}
\label{sec:relevance}

It is evident that the runtime of Ring-BKW is expected to be better than that of standard BKW (in any of its current forms), since the reduction and hypothesis testing phases may be taken to be the same, but the backsubstitution phase is no longer required.  Furthermore, the smaller Ring-LWE problems of hypothesis testing can be solved in parallel.  

Albrecht et al. computed the runtime for BKW \cite{albrecht-bkw}.  This work has been rendered out of date by many of the modern speedups mentioned in the introduction, but it is likely safe to say a few things that still hold true about modern BKW runtimes.  First, the reduction phase is the dominant cost.  Second, however, the backsubstitution phase differs from the reduction phase by a polynomial factor, so eliminating it can be expected to give a polynomial factor speeedup.  

Advanced keying also offers a visible benefit when compared to a ring-blind implementation of BKW.  For, compared to a ring-blind implementation, table sizes are reduced to $1/B$ of their former size and the number of samples used is reduced to approximately $1/n$ as many.  Each sample must be treated rather more carefully however:  it is rotated and a canonical choice made.  However, experiments still indicate increasing runtime gains with dimension, even against traditional BKW with every sample rotated before beginning.  Nevertheless, advanced keying requires block sizes to be a power of $2$, and therefore may or may not be useful or extendable in view of the changing block sizes sometimes employed in BKW reduction.

The Ring-LWE Challenges \cite{challenges} are in the form of \emph{Tweaked Ring-LWE}, which refers to dual Ring-LWE transfered to the unital version (see \cite[\S 2.3]{challenges}), so that the parameter assumptions in this paper apply to the two-power cyclotomic challenges included therein.  It would be very interesting to test these algorithms on those parameters, but it is beyond the scope of this paper.

\bibliographystyle{splncs}
\bibliography{ring-bkw}

\end{document}